\documentclass[]{IEEEtran}
\IEEEoverridecommandlockouts
\usepackage{cite}
\usepackage{amsmath,amssymb,amsfonts}
\usepackage{algorithm}
\usepackage[noend]{algpseudocode}
\usepackage{graphicx}
\usepackage{tabularx}
\usepackage{booktabs}
\usepackage{mathtools}
\usepackage{textcomp}
\usepackage{amsthm}
\usepackage{subcaption}
\usepackage{xcolor}
\def\BibTeX{{\rm B\kern-.05em{\sc i\kern-.025em b}\kern-.08em
    T\kern-.1667em\lower.7ex\hbox{E}\kern-.125emX}}

\usepackage{tikz}
    \usetikzlibrary{shapes.arrows}

\usepackage{pgfplots}
\pgfplotsset{compat=1.15}
\usepackage{adjustbox}
\usepackage{bm}
\usepackage{gincltex}
\usepgfplotslibrary{fillbetween}
\usepgfplotslibrary{groupplots}
\usetikzlibrary{plotmarks}
\usetikzlibrary{patterns}
\pgfplotsset{
tick label style={font=\footnotesize},
label style={font=\footnotesize},
legend style={font=\footnotesize},
}

\definecolor{violet}{rgb}{0.6,0,0.6}%
\definecolor{orange_D}{rgb}{1,0.3,0}%
\definecolor{cyan}{rgb}{0,0.67,0.64}%
\definecolor{red}{rgb}{0.9,0,0}%
\definecolor{yellow}{rgb}{1,0.8,0}

\def \fwidth{0.8\columnwidth}

\def \fheightz {0.5\columnwidth}
\def \fheightt {0.6\columnwidth}

\usepackage{glossaries}    
    
\newacronym[plural=MDPs,firstplural=Markov Decision Processes (MDPs)]{mdp}{MDP}{Markov Decision Process}
\newacronym{iot}{IoT}{Internet of Things}
\newacronym{fec}{FEC}{Forward Error Correction}
\newacronym{snr}{SNR}{Signal to Noise Ratio}
\newacronym{aoi}{AoI}{Age of Information}
\newacronym{paoi}{PAoI}{Peak Age of Information}
\newacronym{qaoi}{QAoI}{Age of Information at Query}
\newacronym{pdf}{PDF}{Probability Density Function}
\newacronym{cdf}{CDF}{Cumulative Distribution Function}
\newacronym{opf}{OPF}{Oldest Packet First}
\newacronym{jfi}{JFI}{Jain Fairness Index}
\newacronym{fcfs}{FCFS}{First Come First Serve}
\newacronym{lcfs}{LCFS}{Last Come First Serve}
\newacronym{pec}{PEC}{Packet Erasure Channel}
\newacronym{irsa}{IRSA}{irregular repetition slotted ALOHA}
\newacronym{fhw}{FHW}{Flatto-Hahn-Wright}
\newacronym{urllc}{URLLC}{ultra-reliable low-latency communication}
\newacronym{mmtc}{mMTC}{massive machine type communications}

\newacronym{pmf}{PMF}{probability mass function}
\newacronym{wgcp}{WGCP}{weighted graph coloring problem}
\newacronym{whcp}{WHCP}{weighted hypergraph coloring problem}
\newacronym{mab}{MAB}{multi-armed bandit}
\newacronym{drl}{DRL}{deep reinforcement learning}
\newacronym{iiot}{IIoT}{Industrial Internet of Things}
\newacronym{lora}{LoRaWAN}{Long Range Wide Area Network}

\newtheorem{theorem}{Theorem}
\DeclareMathOperator*{\argmin}{arg\,min}
\DeclareMathOperator*{\argmax}{arg\,max}

\begin{document}



\title{Learning to Speak on Behalf of a Group: Medium Access Control for Sending a Shared Message} 

\author{Shaan ul Haque, Siddharth Chandak, Federico Chiariotti, Deniz G\"{u}nd\"{u}z, and Petar Popovski\thanks{Shaan ul Haque (shaanhaque2016@gmail.com) and Siddharth Chandak (chandaks@stanford.edu) are with the Department of Electrical Engineering, Indian Institute of Technology Bombay, India. Siddharth Chandak is also with the Department of Electrical Engineering, Stanford University, USA. Federico Chiariotti (fchi@es.aau.dk) and Petar Popovski (petarp@es.aau.dk) are with the Department of Electronic Systems, Aalborg University, Denmark. Deniz G\"{u}nd\"{u}z (d.gunduz@imperial.ac.uk) is with the Department of Electrical and Electronic Engineering, Imperial College London, United Kingdom. This work was partly supported by the Villum Investigator Grant ``WATER'' from the Velux Foundation, Denmark. Deniz G\"{u}nd\"{u}z acknowledges support from the European Research Council (ERC) through Starting Grant BEACON (no. 677854), and the UK EPSRC (project no. EP/T023600/1).}\vspace{-1cm}
}

\maketitle

\begin{abstract}
The rapid development of \gls{iot} technologies has not only enabled new applications, but also presented new challenges for reliable communication with limited resources. In this work, we define a novel problem that can arise in these scenarios, in which a set of sensors need to communicate a joint observation. This observation is shared by a random subset of the nodes, which need to propagate it to the rest of the network, but coordination is complex: as signaling constraints require the use of random access schemes over shared channels, sensors need to implicitly coordinate, so that at least one transmission gets through without collisions. Unlike the majority of existing medium access schemes, the goal is to make sure that the shared message gets through, regardless of the sender. We analyze this coordination problem theoretically and provide low-complexity solutions. While a clustering-based approach is near-optimal if the sensors have prior knowledge, we provide a distributed \gls{mab} solution for the more general case and validate it by simulation.
\end{abstract}

\glsresetall


\section{Introduction}
\label{sec:intro}
\glsresetall

Over the past few years, the rise of the \gls{iot}~\cite{cheng2018industrial} has opened new possibilities in the manufacturing, energy, and health sectors. The promise of 5G and beyond networks is to support massive numbers of sensors and machine-type devices, along with sporadic low-latency communications, without affecting human communication traffic~\cite{gangakhedkar2018use}. However, there are still many open problems in coordinating medium access for sporadically active sensors~\cite{vilgelm2021random}, particularly in remote deployments with very limited resources. As random access schemes like slotted ALOHA are entirely distributed, and only require limited overhead for slot synchronization, they often represent the best choice for these scenarios~\cite{jian2016random}, but the risk of collisions is a significant drawback, particularly for scenarios with a large number of nodes~\cite{saadawi1981analysis}.

An interesting scenario in this context is the transmission of an observation or a decision that is shared by the devices in a given \emph{active set}. For example, this could be the position of a target object, or an abnormal value of a parameter, which is shared by a subset of the sensors in the network and needs to be communicated to the rest of the sensor network or to an external controller. This scenario, which we call \emph{medium access with a shared message}, is relevant in several applications pertaining to networked control and coordination, in which the whole network needs to perform an action synergistically. In our case, each of the active sensor nodes has the same piece of information, namely the \emph{shared message}, which they want to deliver to the other nodes in a single shot, i.e., without the possibility of retransmissions or coordination~\cite{liu2019wireless}. This is relevant both in \gls{urllc}, in which the strictest constraint is the latency requirement, and in wide-area scenarios, in which energy consumption is the foremost concern.
In this context, traditional throughput maximization-based schemes are extremely inefficient, as they assume each packet carries independent information. We further assume that the active set evolves in a random fashion from one time slot to the next, and each sensor knows only its own membership of the current active set. The challenge here is that a node from the active set cannot coordinate with the other active nodes to send the shared message. The problem looks deceptively simple, but coordinating with limited signaling or shared prior knowledge is extremely difficult. Exploiting correlations in the activity of the sensors to maximize the throughput has been considered in ~\cite{kalor2018random, ali2020sleeping}. A related, but different problem has been treated in~\cite{CommonAlarm19}, where the objective is to reliably transmit a shared alarm message by superposing individual signals; this is different from the collision model adopted here, and does not require coordination between the sensors.

We first prove the existence of an optimal deterministic solution to this coordination problem, which is then shown to be NP-hard by modeling the correlations in sensor activity as edge weights in a graph. The coordination problem with deterministic strategies is an instance of a slightly modified \gls{wgcp}~\cite{chang2010weighted}, which is a well-known NP-hard problem. Distributed solutions to the \gls{wgcp} exist~\cite{barenboim2013distributed}, and have been used in communications scenarios~\cite{hernandez2014frogsim}, but they either require signaling between the nodes or more extensive shared feedback. Our scenario allows for extremely limited signaling, as the sensors only receive an acknowledgment (ACK) in case of correct packet reception. We provide a clustering-based solution, which is near-optimal if only two sensors are active at a given time. However, this does not generalize to larger active sets. Instead, we introduce
a heuristic to construct the clusters in the general case. We also propose an efficient distributed learning solution using \glspl{mab}, which can learn correlation patterns and adapt to them without any signaling except for the ACK. 


The rest of the letter is organized as follows: Sec.~\ref{sec:system} presents the system model and theoretical analysis. Sec.~\ref{sec:solution} presents our two heuristic solutions, which are evaluated through numerical simulations in Sec.~\ref{sec:results}. Finally, we conclude the paper and present ideas for future work in Sec.~\ref{sec:concl}.

\section{System model}
\label{sec:system}

We consider a set $\mathcal{N}$ of $N$ wireless sensors, which share $M$ orthogonal transmission opportunities in time or frequency. In each time slot $t$, a random subset $\mathcal{A}(t)\subseteq\mathcal{N}$ of sensors, with cardinality $A(t) = |\mathcal{A}(t)|$, become active. The active set is drawn independently at each slot according to \gls{pmf} $p_A(\mathcal{A})$. A shared message, e.g., an alarm signal, is to be transmitted by the active sensors to the inactive nodes (which cannot sense the state of the system) or to an external controller. The objective of the active sensors is to deliver the message over the $M$ opportunities, regardless of which sensor it comes from. The challenge lies in the fact that the active sensors have no knowledge of the other sensors in $\mathcal{A}$, and there is no way to explicitly coordinate. The sensors can only agree on a MAC protocol \textit{a priori}. 

It is clear that the inactive sensors at each time slot must remain silent. Each active sensor $a\in\mathcal{A}(t)$ must decide on a transmission pattern, called a \emph{move}, expressed as vector $\mathbf{x}_a\in\{0,1\}^M$, where $x_{a,m}=1$ means that sensor $a$ transmits at the $m$-th opportunity. We can represent the moves of all the sensors as an $A(t)\times M$ matrix $\mathbf{X}(t)$, with vector $\mathbf{x}_a$ as its $a$-th row.  In the following, we omit the time index $t$ for readability. We consider a simple collision channel, where the condition for transmission success $\xi(\mathbf{X},\mathcal{A})$ is:
\begin{equation}
  \xi(\mathbf{X},\mathcal{A})=I\left(\exists m\in\{1,\ldots,M\}:\sum\limits_{a\in\mathcal{A}} x_{a,m}=1\right),\label{eq:move_value}
\end{equation}
where $I(\cdot)$ is the indicator function, equal to 1 if the condition holds, and 0 otherwise.
Note that there is a total of $2^M$ possible moves for each sensor at each time slot. The constrained problem in which each sensor can transmit only once is a special case of this problem, and in general leads to suboptimal solutions, including in some of the cases that we analyze below. The basic rationale for considering multiple transmissions from the same sensor is the same as in \gls{irsa} schemes~\cite{munari2021modern}, which exploit the repetitions to improve throughput and reliability. We can define the \emph{strategy} of node $a$ as the \gls{pmf} $\phi_a(\mathbf{x})$ over the set of possible moves, and represent the strategies of all the sensors in matrix $\Phi\in [0,1]^{N\times 2^M}$, where element $\phi_{n\ell}$ corresponds to the probability of sensor $n$ choosing move $\ell$ when it is active. We have $\sum_{x=1}^{2^M} \phi_{n,x} = 1$, $\forall n \in\mathcal{N}$. By applying the law of total probability, we get:
\begin{equation}
  \mathbb{E}\left[\xi|\Phi\right]=\sum_{\mathcal{A}\in\mathcal{P}(\mathcal{N})}p_A(\mathcal{A})\sum_{\mathclap{\mathbf{X}\in\{0,1\}^{N\times M}}}\xi(\mathbf{X},\mathcal{A})\prod_{a \in \mathcal{A}} \phi_a(\mathbf{x}_a), \label{eq:strategy_value}
\end{equation}
where $\mathcal{P}(\cdot)$ denotes the power set, and $\phi_a(\mathbf{x}_a)$ denotes the probability of sensor $a$ choosing move $\mathbf{x}_a$. We then define our optimization problem, whose solution gives $\Phi^*$, one of the strategies that maximize the expected delivery probability:
\begin{equation}
  \Phi^*=\argmax_{\Phi\in[0,1]^{N\times 2^M}} \mathbb{E}\left[\xi|\Phi\right].
\end{equation}

\begin{theorem}\label{th:det_solution}
At least one of the optimal solutions to the optimization problem is a deterministic strategy; that is,
  \begin{equation}
  \exists \ \Phi\in\{0,1\}^{N\times 2^M}:\mathbb{E}\left[\xi|\Phi\right]=\mathbb{E}\left[\xi|\Phi^*\right].\label{eq:opt_cont}
\end{equation}
\end{theorem}
\begin{proof}
  As the value of $\xi$ is between $0$ and $1$, its expected value is bounded in the compact interval $[0,1]$, and the $[0,1]^{N\times 2^M}$ region specified by the constraints on $\Phi$ is also compact. Accordingly, there exists at least one global maximum. 
  
 Assume that in the optimal solution there is at least one sensor with a non-deterministic policy, which we denote as 1. We can then look at all possible moves $\mathbf{x}_1$ and compute the values of the deterministic strategies for sensor 1:
  \begin{equation}
  \begin{aligned}
    \mathbb{E}\left[\xi|(\delta(\mathbf{x}_1);\Phi^*_{-1})\right]=&\sum_{\mathclap{\mathcal{A}\in\mathcal{P}(\mathcal{N}):1\notin\mathcal{A}}}p_A(\mathcal{A})\mathbb{E}\left[\xi|\Phi,\mathcal{A}\right]+\sum_{\mathclap{\mathcal{A}\in\mathcal{P}(\mathcal{N}):1\in\mathcal{A}}}p_A(\mathcal{A})\\
    &\times\sum_{\mathclap{\mathbf{X_{-1}}\in\{0,1\}^{(N-1)\times M }}}\xi((\mathbf{x}_1;\mathbf{X}_{-1}),\mathcal{A})\prod_{a=2}^A \phi_a(\mathbf{x}_a),
    \end{aligned}
  \end{equation}
  where $\mathbf{X}_{-1}$ is matrix $\mathbf{X}$ without its first row. We can then substitute this into~\eqref{eq:strategy_value}:
  \begin{equation}
     \mathbb{E}\left[\xi|\Phi\right]=\sum_{\mathclap{\mathbf{x}_1\in\{0,1\}^{M}}}\phi_1(\mathbf{x}_1)\mathbb{E}\left[\xi|(\delta(\mathbf{x}_1);\Phi^*_{-1})\right].
  \end{equation}
  Then, the strategy $\Phi'$ that maximizes the expected value:
  \begin{equation}
    \Phi'=\left(\argmax_{\mathbf{x}_1\in\{0,1\}^{M}} \mathbb{E}\left[\xi|(\delta(\mathbf{x}_1;\Phi^*_{-1})\right];\Phi_{-1}\right),
  \end{equation}
  will be a deterministic one due to the linearity of the objective with respect to $\phi_1(\mathbf{x}_1)$ and the compactness of the simplex. By repeating this operation for all the sensors with non-deterministic strategies, we can find a deterministic solution $\Phi^{(d)}$ that satisfies the condition in~\eqref{eq:opt_cont}.
\end{proof}

Hence, we can focus on deterministic strategies over the discrete set of moves $\{0,1\}^{N\times M}$ instead of the continuous probability space $[0,1]^{N\times 2^M}$ without loss of optimality:
\begin{equation}
  \mathbf{X}^*=\argmax_{\mathbf{X}\in\{0,1\}^{N\times M}} \mathbb{E}\left[\xi|\mathbf{X}\right].\label{eq:opt}
\end{equation}
The problem is trivial if only one user is active at a time, i.e., if $A(t)=1, \forall t$: in that case, there is no interference and the trivial solution $\mathbf{x}_{a,m}=1, \forall a\in\mathcal{A}(t)$ is always successful. The same happens if $M\geq N$, in which case the access to the medium can be entirely orthogonal. However, the problem is extremely complex in the general case.

\begin{theorem}
  The problem defined in~\eqref{eq:opt} is NP-hard if $A\geq 2,\ \forall\ \mathcal{A}: p_A(\mathcal{A})>0$.
\end{theorem}

\begin{proof}
  We will prove the NP-hardness of the problem if all the active sets have $A(t)=2,\ \forall t$, i.e., if exactly two nodes are active at a given time, by showing its equivalence to an instance of the \gls{wgcp}~\cite{hell1990complexity}. \gls{wgcp} determines the $k$-coloring of a weighted undirected graph $\mathcal{G}=(\mathcal{V},\mathcal{E},w)$ with minimum weight~\cite{chang2010weighted}: it assigns an integer number $c_v\in\{1,\ldots,k\}$ to each vertex $v\in \mathcal{V}$. The optimal $k$-coloring is then the solution to the following weight minimization:
  \begin{equation}
    \mathbf{c}^*=\argmin_{\mathbf{c}\in\{1,\ldots,k\}^{|\mathcal{V}|}}\sum_{(u,v)\in \mathcal{E}:c_u=c_v}w_{u,v}.
  \end{equation}
In our case, we can consider a fully connected graph with $\mathcal{V}=\mathcal{N}$. We assign weights equivalent to the probability of two sensors being active at the same time, i.e., $w_{u,v}=p_A(\{u,v\})$. We can then define the weight minimization as:
\begin{equation}
      \mathbf{X}^*=\argmin_{\mathbf{X}\in\{0,1\}^{N\times M}}\sum_{u,v \in \mathcal{N}, u \neq v} p_A(\{u,v\})\xi(\mathbf{X},\{u,v\}).\label{eq:wgcp}
\end{equation}
As solving the communication problem is equivalent to solving the \gls{wgcp}, it is NP-hard for $A=2$. If there are sets with non-zero probability and size $A>2$, the problem is equivalent to a \gls{whcp}, which models active sensor sets as weighted edges between two or more nodes and is also NP-hard~\cite{dinur2005hardness}. Although the problem is equivalent to a \gls{whcp} in terms of complexity, the definition is slightly different, as some combination of strategies might result in a successful transmission even if multiple sensors choose the same move, e.g., if two nodes choose to be silent, while the third transmits. The condition is then on $\xi$, as in~\eqref{eq:wgcp}, and not on the colors (i.e., the strategies) themselves.
\end{proof}

\section{Solving the coordination problem}\label{sec:solution}

In the following, we present two solutions to this complex problem. The first solution is based on clustering and requires full knowledge of $p_A$. The second is based on \gls{mab} learning, and requires a training period in which the nodes attempt to communicate and learn how to coordinate. Unlike the clustering-based solution, \gls{mab} learning can be performed online, as it does not require an oracle view of the network: while knowing exactly which sensors are active at any given time is necessary to estimate $p_A$, the sensors can independently implement \glspl{mab} and try all strategies, using acknowledgments for correctly transmitted packets as their only feedback. The clustering-based solution is limited to the case with $A=2$, and its generalization to larger active sets is non-trivial. Yet, its main advantage is its immediate applicability when the activation distribution $p_A$ is known, with no training period required, as it is derived analytically.

Although finding the optimal strategy is NP-hard, we can always find an optimal solution by brute force, enumerating all $N^{2^M}$ possible strategies. Naturally, the complexity of this solution makes it impractical in most cases, but we can still perform the computation for the small networks that we consider, in order to show the optimality gap.

\subsection{Clustering-based solution}
When $A=2$, we can use the graph representation that we exploited to prove the NP-hardness of the problem to design a clustering-based solution. In this solution, we will group the sensors into $2^M$ clusters, and assign the same move to all the sensors in the same cluster. We will have a collision if and only if two sensors from the same cluster are active simultaneously. This approach is suboptimal, but the optimality gap is small in the scenarios we considered. 

To minimize the collision probability, we employ the probability of two sensors being together in the active set as a cost $d_{u,v}=p_A(\{u,v\})$, where the total cost of a cluster $\mathcal{C}$ is given by the sum of the pairwise costs in the cluster:
\begin{equation}
  d(\mathcal{C})=\sum_{u, v\in\mathcal{C}, u < v} p_A(\{u,v\}).
\end{equation}
Note that the total cost summed over all clusters is equivalent to the failure probability of the scheme if $|\mathcal{A}|=2$, i.e., only one pair of sensors can be active at any time. 

We employ a divisive clustering approach~\cite{roux2018comparative}, which tries to minimize this total cost at each step in a greedy fashion. This approach significantly outperformed agglomerative and K-means clustering in our experiments. We use a modified version of the DIANA clustering algorithm~\cite{kaufman2009finding}, which starts from a single cluster, then iteratively splits the cluster by dividing the highest-cost cluster, starting from the highest-cost node in it. If $A>2$, we need to consider nodes one at a time, starting from the highest-probability one and assigning each node to a strategy iteratively. This heuristic is relatively simple and usually has a small optimality gap, but also requires full knowledge of the activation probability matrix.

\begin{figure*}[t]
    \begin{subfigure}[b]{.21\linewidth}
	    \centering
	    \resizebox{.7\columnwidth}{!}{
        \begin{tikzpicture}
  
	\draw[line width=0.1cm] (0,0) circle (1.35cm);
	
    \draw[fill=blue] (0,0) -- +(72:1.3) arc (72:108:1.3) -- (0,0);
    \draw[fill=blue!55] (0,0) -- +(36:1.3) arc (36:72:1.3) -- (0,0);
    \draw[fill=blue!55] (0,0) -- +(108:1.3) arc (108:144:1.3) -- (0,0);
    \draw[fill=blue!25] (0,0) -- +(0:1.3) arc (0:36:1.3) -- (0,0);
    \draw[fill=blue!25] (0,0) -- +(144:1.3) arc (144:180:1.3) -- (0,0);
    \draw[fill=blue!15] (0,0) -- +(-36:1.3) arc (-36:0:1.3) -- (0,0);
    \draw[fill=blue!15] (0,0) -- +(180:1.3) arc (180:216:1.3) -- (0,0);
    \draw[fill=blue!5] (0,0) -- +(-72:1.3) arc (-72:-36:1.3) -- (0,0);
    \draw[fill=blue!5] (0,0) -- +(216:1.3) arc (216:252:1.3) -- (0,0);
    
%

  	\foreach \i in {1,2,...,10}{
		\def\angle{-\i*36+126}
		\draw[thin] (\angle:1.2cm) -> (\angle:1.3cm);
		\node at (\angle:1cm) {\large\i};
	};

\end{tikzpicture}

\vspace{0.4cm}}
        \caption{Graphical representation.}
        \label{fig:corr_sc}
    \end{subfigure}	
    \begin{subfigure}[b]{.37\linewidth}
	    \centering
        \input{./plots/thompson/corr2_10.tex}
        \caption{Performance with $A(t)=2, \forall t$.}
        \label{fig:corr_2}
    \end{subfigure}	
	\begin{subfigure}[b]{.37\linewidth}
	    \centering
        \input{./plots/thompson/corr3_10.tex}
        \caption{Performance with $A(t)=3, \forall t$.}
        \label{fig:corr_3}
    \end{subfigure}	
     \caption{Results in the regular scenario with $N=10$, $M=2$.\vspace{-0.3cm}}
 \label{fig:corr}
\end{figure*}

\subsection{Distributed learning approach}
It is also possible to learn the optimal policy in a distributed fashion, implementing a \gls{mab} for each sensor. \glspl{mab} are learning agents that have a number of \emph{arms}, which correspond to possible actions, and learn by trying each arm and estimating its expected value over time. There are several sampling strategies to choose which action to use, balancing between exploration (i.e., choosing the action to gain new information) and exploitation (i.e., choosing the action with the highest potential value based on past experience). In our problem, each sensor chooses from $2^M$ arms, each corresponding to a different move. Each sensor will try different transmission patterns whenever it is active, and get a \emph{reward} $\xi(t)=1$ if the communication is successful and 0 otherwise, i.e., the shared reward. Since all active sensors are trying to communicate the same piece of information, the communication is considered as successful even if the packet is transmitted by another sensor; that is, the reward depends on the network's success as a whole, and not any individual sensor's action.

Our solution is a distributed version of Thompson sampling for Bernoulli \glspl{mab}~\cite{ghalme2017thompson}: every time a sensor is active, it chooses a move based on the Thompson sampling algorithm, which is described and explained in~\cite{russo2018tutorial}, and observes the reward. The sampling strategy needs to be entirely distributed, as the only coordination signal available to the agents is the shared binary feedback. The system can also be modeled as a repeated $N$-player cooperative incomplete information game, in which the unknown information is the membership of the active set. Theorem~\ref{th:det_solution} proves that a pure Nash equilibrium exists, and that it represents the optimal strategy, but there might be other suboptimal equilibria. The Thompson sampling solution converges to an equilibrium with bounded regret~\cite{asghari2020regret}, in which no agent tries to unilaterally deviate from the joint policy, but the distributed solution might converge to a suboptimal equilibrium, i.e., a local maximum.

The distributed Thompson sampling algorithm is given in Algorithm~\ref{alg:thompson}, which is run independently at each sensor: every time a node is active, it uses a Beta distribution to assign probabilities to each action based on the expected reward of the action, then observes the result of the move and updates its values. This requires no coordination, as each sensor acts independently, and leads to quick convergence.

\begin{algorithm}[t]
\caption{Distributed Thompson sampling}\label{alg:thompson}
\scriptsize
\begin{algorithmic}[1]
\State$\alpha\gets$ ones($2^M$)
\State$k\gets$ones($2^M$)
\While{sampling} \Comment{Loop over communication rounds}
\If{active} \Comment{Only transmit if active}
\State $X\gets$\Call{chooseAction}{$M,\alpha,k$}
\State $\xi\gets$\Call{move}{X}\Comment{Observe result of the action}
\State $\alpha,k\gets$\Call{ThompsonUpdate}{$X,\xi,\alpha,k$}\Comment{Update action value}
\EndIf
\EndWhile
\Function{ThompsonUpdate}{$X,\xi,\alpha,k$}
\State $k[X]\gets k[X]+1$ \Comment{Increase sample size for $a$}
\State $\alpha[X]\gets\alpha[X]+\xi$ \Comment{Update total reward}
\State \textbf{return} $\alpha,k$
\EndFunction
\Function{chooseAction}{$M,\alpha,k$}
\State $\theta\gets$zeros($2^M$)
\For{$X\in\{0,\ldots,2^M-1\}$}
\State $\theta[X]\gets$\Call{Beta}{$\alpha[X],k[X]+1-\alpha[X]$} \Comment{Sample from Beta distribution}
\EndFor
\State \textbf{return} $\argmax(\theta)$
\EndFunction
\end{algorithmic}
\end{algorithm}

\section{Numerical results}
\label{sec:results}

In the following, we show the simulation results for the clustering and learning solutions. We consider a system with $N$ sensor nodes, transmitting over $M=2$ opportunities. We limit the number of simultaneously active agents to the cases with $A=2$ and $A=3$. Although the clustering solution only works in the former, the \gls{mab} solution works equally well in both of these cases. Furthermore, we consider two different types of node activation distributions:
\begin{enumerate}
 \item \emph{Regular activation}: we consider a correlated activation pattern with significant regularity, which might be, e.g., due to the physical location of the sensors. Sensors that are closer together are often simultaneously active, while sensors that are farther apart have a smaller joint activation probability;
 \item \emph{General case}: $p_A$ is a general \gls{pmf} with no apparent regularity. Hand-designing a solution for this case is extremely difficult, as it requires to solve the overall problem and there are no regular features to exploit.
\end{enumerate}

We also tested the algorithm in a case with deterministic activation, i.e., in which each sensor is only in one possible active set. As expected, this case leads to the election of a \emph{leader} in each group, who always transmits, while the other nodes remain silent. The sensors as a whole can transmit 100\% of the alarms in this case, and convergence of the distributed \gls{mab} to the optimal solution is extremely fast. 

\begin{figure*}
    \begin{subfigure}[b]{.32\linewidth}
	    \centering
        \input{./plots/thompson/gen2_20.tex}
        \caption{$M=2,\ A(t)=2, \forall t$.}
        \label{fig:gen_2}
    \end{subfigure}	
	\begin{subfigure}[b]{.32\linewidth}
	    \centering
        \input{./plots/thompson/gen3_20.tex}
        \caption{$M=2,\ A(t)=3, \forall t$.}
        \label{fig:gen_3}
    \end{subfigure}	
    	\begin{subfigure}[b]{.32\linewidth}
	    \centering
        \input{./plots/thompson/test2.tex}
        \caption{$M=4,\ A(t)=4, \forall t$.}
        \label{fig:m=4}
    \end{subfigure}	
     \caption{Performance in the general scenario with $N=20$.\vspace{-0.3cm}}
 \label{fig:gen}
\end{figure*}

\subsection{Regular activation}

In this model, we assume that $N=10$ sensors are located around a circle (see Fig. \ref{fig:corr_sc}), and that sensors closer to each other are more likely to be active at the same time. At each time $t$, first sensor in $\mathcal{A}(t)$ is picked at random with probability $\frac{1}{N}$, while the second sensor is picked with probability $p(d)$, which depends on the distance from the first sensor along the circle: if $d\!=\!1$, the probability is 0.275, if $d\!=\!2$, it is 0.125, if $d\!=\!3$, it is 0.075, and if $d\!=\!4$, it is 0.025. This is shown graphically in Fig.~\ref{fig:corr_sc}: if sensor 1 is picked, the higher probability of picking closer sensors as the second sensor is depicted as a deeper blue. The sensor diametrically opposite to the first one is never picked. When $A=3$, the third sensor is picked from the same distribution, considering the distance from the second sensor. The strategy is intuitive: nodes that often appear together should have different strategies, so as to avoid collisions and transmit the shared message in at least one of the 2 opportunities. This implies that some sensors are always silent, and others use both opportunities: as this is only a problem in case of two sensors with the same strategy, maximizing the number of different strategies by assigning $(0,0)$ and $(1,1)$ to some nodes is the best choice. If $A=3$, the situation is slightly different, as it is possible to achieve an optimal solution by only using strategies $(1,0)$ and $(0,1)$: the objective of the sensors in this case is to avoid a scenario in which all three collide in the same transmission opportunity, as any other combination leads to a successful transmission. 

The results are shown in Fig.~\ref{fig:corr}, which includes the best and worst results from 5 independent runs of the distributed \gls{mab} approach: in the worst case, this approach might get stuck in a suboptimal equilibrium, corresponding to a local maximum of the reward function. When $A=2$, the \gls{mab} solution always reaches the optimum, while the clustering solution has a small optimality gap. If $A=3$, the situation is reversed, and the \gls{mab} solution converges to a slightly suboptimal solution. In the worst case, the optimality gap of the worst learning curve is close to 0.01. However, the \gls{mab} solution does not require any prior knowledge of the activation probabilities, and can be trained online with no knowledge beyond a shared ACK signal in relatively few iterations. 

\subsection{General scenario}

Finally, we show the results for a general scenario, using a randomly drawn activation probability matrix with $N=20$. In this scenario, $A=3$ case is harder, as Fig.~\ref{fig:gen} shows: the success rate for the best \gls{mab} training is slightly below $0.75$, while still outperforming the heuristic. If $A=2$, the \gls{mab} and clustering solutions have similar performance, close to 0.85. In this case, we do not know the optimal solution as the brute-force search required to find it is too complex due to the large number of sensors. However, training times for the \gls{mab} solution appear to scale, and the final performance is close to the clustering-based heuristics, which starts with the full knowledge of the environment. The \gls{mab} solution with Thompson sampling is general, fast, efficient, and robust, as it can easily deal with errors on the acknowledgments as well as the packets. However, Thompson sampling does not always reach the optimal solution, as it can get stuck in local maxima: using an $\varepsilon$-greedy strategy guarantees better performance after convergence, but requires approximately 1000 times more samples to reach convergence, and as such becomes impracticable in realistic scenarios.

The same happens in a more complex scenario, shown in Fig.~\ref{fig:m=4}, with $M=4$ and $A=4$: in this case, the \gls{mab} solution performs slightly worse than the heuristic, but it can still reach convergence in relatively few samples, and the optimality gap of the worst learning curve is just 0.03. 
It is important to note that the constructive heuristic requires full knowledge of the activation probabilities, which are hard to estimate accurately in practice. On the other hand, the \gls{mab} solution only requires the ACK signal, and has relatively fast convergence, making it an attractive alternative.

\section{Conclusions and future work}
\label{sec:concl}

We have introduced and examined a novel distributed coordination and communication problem, in which multiple wireless sensors need to coordinate and find a common strategy to transmit a shared message. The problem is fundamentally different from general error rate minimization, in which every packet contains different information, and is NP-hard. We proved that a deterministic optimal solution exists, and proposed two heuristics with a small optimality gap in practical conditions. The clustering-based heuristic is close to the optimum, but requires \emph{a priori} knowledge of the system, while the \gls{mab} solution can get stuck in a local maximum, but it can be trained online with no additional signaling.

There are several possible avenues of future work on the topic, including the use of more advanced learning mechanisms, such as neural network-based bandits which can generalize experience and converge with fewer training samples. Another interesting research direction is the application of these principles to swarm control, in which the sensors are exchanging information about a shared environment, which they can directly modify by acting in concert. Finally, the scalability of the solutions to scenarios with massive numbers of devices is a significant challenge for future research.

\bibliographystyle{IEEEtran}
\bibliography{references}
\end{document}